\documentclass[11pt,letterpaper]{article}

\usepackage{amsmath,amsfonts,amsthm,amssymb,relsize,bm}  
\theoremstyle{plain}

\numberwithin{equation}{section}

\newtheorem{thm}{Theorem}[section]
\newtheorem{lem}[thm]{Lemma}
\newtheorem{cor}[thm]{Corollary}

\theoremstyle{definition}  
\newtheorem{exam}{Example}  

\allowdisplaybreaks  

\newcommand{\tbullet}{\mathrel{\raise .4ex\hbox{\tiny$\bullet$}}} 

\newcommand{\rmtr}{\mathrm{tr\,}}
\newcommand{\rmin}{\mathrm{In\,}}
\newcommand{\rmpost}{\mathrm{Post\,}}

\newcommand{\ascript}{\mathcal{A}}
\newcommand{\cscript}{\mathcal{C}}
\newcommand{\escript}{\mathcal{E}}
\newcommand{\iscript}{\mathcal{I}}
\newcommand{\jscript}{\mathcal{J}}
\newcommand{\lscript}{\mathcal{L}}
\newcommand{\oscript}{\mathcal{O}}
\newcommand{\pscript}{\mathcal{P}}
\newcommand{\sscript}{\mathcal{S}}

\newcommand{\iscripthat}{\widehat{\iscript}}
\newcommand{\jscripthat}{\widehat{\jscript}}

\newcommand{\brac}[1]{\left\{#1\right\}}
\newcommand{\paren}[1]{\left(#1\right)}
\newcommand{\sqbrac}[1]{\left[#1\right]}
\newcommand{\parensq}[1]{{\left(#1\right]}} 
\newcommand{\elbows}[1]{{\left\langle#1\right\rangle}}

\begin{document}

\title{COMBINATIONS OF\\QUANTUM OBSERVABLES\\AND INSTRUMENTS}
\author{Stan Gudder\\ Department of Mathematics\\
University of Denver\\ Denver, Colorado 80208\\
sgudder@du.edu}
\date{}
\maketitle

\begin{quote}
\hskip 1pc\textit{``You should conduct research of such a high quality that \newline people remember your name.''}
{\hfill ---{\small Author, Unknown}}
\end{quote}
\vskip 2pc

\begin{abstract}
This article points out that observables and instruments can be combined in many ways that have natural and physical interpretations. We shall mainly concentrate on the mathematical properties of these combinations. Section~1 reviews the basic definitions and observables are considered in Section~2. We study parts of observables, post-processing, generalized convex combinations, sequential products and tensor products. These combinations are extended to instruments in Section~3. We consider properties of observables measured by combinations of instruments. We introduce four special types of instruments, namely Kraus, L\"uders, trivial and semitrivial instruments. We study when these types are closed under various combinations. In this work, we only consider finite-dimensional quantum systems. A few of the results presented here have appeared in the author's previous articles.
\cite{gud220,gud320,gud420}.
\end{abstract}

\section{Basic Definitions}  
Let $\lscript (H)$ be the set of linear operators on a finite-dimensional complex Hilbert space $H$. For $S,T\in\lscript (H)$ we write $S\le T$ if
$\elbows{\phi ,S\phi}\le\elbows{\phi ,T\phi}$ for all $\phi\in H$. We define the set of \textit{effects} by
\begin{equation*} 
\escript (H)=\brac{a\in\lscript (H)\colon 0\le a\le I}
\end{equation*}
where $0,I$ are the zero and identity operators, respectively. The effects correspond to yes-no experiments and $a\in\escript (H)$ is said to \textit{occur} when a measurement of $a$ results in the outcome yes. We call $\rho\in\escript (H)$ a \textit{partial state} if $\rmtr (\rho )\le 1$ and
$\rho$ is a \textit{state} if $\rmtr (\rho )=1$. We denote the set of partial states by $\sscript _p(H)$ and the set of states by $\sscript (H)$. If
$\rho\in\sscript (H)$, $a\in\escript (H)$, we call $\pscript _\rho (a)=\rmtr (\rho a)$ the \textit{probability that} $a$ \textit{occurs} in the state $\rho$ \cite{bgl95,hz12,kra83,lah03}.

We denote the unique positive square-root of $a\in\escript (H)$ by $a^{1/2}$. For, $a,b\in\escript (H)$, their \textit{sequential product} is the effect $a\circ b=a^{1/2}ba^{1/2}$, where $a^{1/2}ba^{1/2}$ is the usual operator product \cite{gg02,gn01}. We interpret $a\circ b$ as the effect that results from first measuring $a$ and then measuring $b$. Let $\Omega _A$ be a finite set. A \textit{finite observable} \cite{hz12,nc00} with
\textit{outcome space} $\Omega _A$ is a subset
\begin{equation*} 
A=\brac{A_x\colon x\in\Omega _A}\subseteq\escript (H)
\end{equation*}
such that $\sum\limits _{x\in\Omega _A}A_x=I$. We denote the set of finite observables on $H$ by $\oscript (H)$. In the sequel, an observable will always mean a finite-observable. We interpret $A\in\oscript (H)$ as a measurement with possible outcomes $x\in\Omega _A$ and $A_x$ is the effect that occurs when the measurement result is $x$. If $A\in\oscript (H)$, we define the effect-values measure $X\mapsto A_X$ from
$2^{\Omega _A}$ to $\escript (H)$ by $A_X=\sum\limits _{x\in X}A_x$. The \textit{distribution of} $A\in\oscript (H)$ \textit{in the state} $\rho\in\sscript (H)$ is defined by $\Phi _\rho ^A(x)=\rmtr (\rho A_x)$ for all $x\in\Omega _A$. Then
\begin{equation*} 
\Phi _\rho ^A(X)=\sum _{x\in X}\Phi _\rho ^A(x)
\end{equation*}
gives the probability that $A$ has an outcome in $X\subseteq\Omega _A$ when the system is in the state $\rho$. Notice that
$X\mapsto\Phi _\rho ^A(X)$ is a probability measure on $\Omega _A$.

An \textit{operation} on $H$ is a completely positive, trace-reducing, linear map $\ascript\colon\lscript (H)\to\lscript (H)$
\cite{bgl95,hz12,kra83,nc00}. Trace-reducing implies that $\ascript\colon\sscript _p(H)\to\sscript _p(H)$. According to Kraus' Theorem
\cite{hz12,kra83,nc00} every operation $\ascript$ has the form $\ascript (T)=\sum\limits _{i=1}^nS_iTS_i^*$ where $S_i\in\lscript (H)$ satisfy
$\sum\limits _{i=1}^nS_i^*S_i\le I$. An operation $\ascript$ is a \textit{channel} if $\ascript (\rho )\in\sscript (H)$ for all $\rho\in\sscript (H)$
\cite{hz12,kra83}. In this case, the \textit{Kraus operators} $S_i$ satisfy $\sum\limits _{i=1}^nS_i^*S_i=I$. We denote the set of channels on $H$ by $\cscript (H)$. For a finite set $\Omega _\iscript$, a \textit{finite instrument} with \textit{outcome space} $\Omega _\iscript$ is a set of operations $\iscript =\brac{\iscript _x\colon x\in\Omega _\iscript}$ such that $\sum\limits _{x\in\Omega _\iscript}\iscript _x\in\cscript (H)$
\cite{bgl95,gud220,gud320,hz12,lah03}. Defining $\iscript _X$ for $X\subseteq\Omega _\iscript$ by $\iscript _X=\sum\limits _{x\in X}\iscript _x$, we see that $X\mapsto\iscript _X$ is an operation -valued measure on $H$. We denote the set of finite instruments on $H$ by $\rmin (H)$. The \textit{distribution} of $\iscript\in\rmin (H)$ \textit{in the state} $\rho\in\sscript (H)$ is defined by
$\Phi _\rho ^\iscript (x)=\rmtr\sqbrac{\iscript _x(\rho )}$ for all $x\in\Omega _\iscript$. Then
\begin{equation*} 
\Phi _\rho ^\iscript (X)=\sum _{x\in X}\Phi _\rho ^\iscript (x)
\end{equation*}
gives the probability that $\iscript$ has an outcome in $X$ when the system is in the state $\rho$. As with observables,
$X\mapsto\Phi _\rho ^\iscript (X)$ gives a probability measure on $\Omega _\iscript$. If $A\in\oscript (H)$, we say that an instrument
$\iscript\in\rmin (H)$ \textit{measures} $A$ (or is \textit{compatible with} $A$) if $\Omega _\iscript =\Omega _A$ and
$\Phi_\rho ^\iscript (x)=\Phi _\rho ^A(x)$ for all $x\in\Omega _A$, $\rho\in\sscript (H)$ \cite{bgl95,hz12,lah03}. This condition is equivalent to
$\rmtr (\rho A_X)=\rmtr\sqbrac{\iscript _X(\rho )}$ for all $X\subseteq\Omega _A$, $\rho\in\sscript (H)$.

If $\iscript\in\rmin (H)$, there exists a unique $\iscripthat\in\oscript (H)$ such that $\iscript$ measures $\iscripthat$ \cite{hz12}. However, an observable has many instruments that measure it. We view $\iscript\in\rmin (H)$ as an apparatus that can be employed to measure the observable $\iscripthat\in\oscript (H)$. However, $\iscript$ gives more information than $\iscripthat$ because $\iscript _x(\rho )\in\sscript _p(H)$ updates the state $\rho$ when the outcome $x$ is observed. There is no corresponding unambiguous updating for observables.

\section{Observables}  
This section discusses functions of observables and various combinations of observables. If $A\in\oscript (H)$ and
$f\colon\Omega _A\to\Omega$ is a surjection, we define $f(A)\in\oscript (H)$ to have outcome space $\Omega$ and for every $y\in\Omega$
\begin{equation*} 
f(A)_y=A_{f^{-1}(y)}=\sum _x\brac{A_x\colon f(x)=y}
\end{equation*}
We say that the observable $f(A)$ is \textit{part} of the observable $A$ \cite{fhl18,gud420,hrsz09,hmr14}. As its name suggests, we think of
$f(A)$ as an observable that measures only a part of $A$. Two observables $A,B\in\oscript (H)$ are said to \textit{coexist} if there exists a
$C\in\oscript (H)$ and surjections $f\colon\Omega _C\to\Omega _A$, $g\colon\Omega _C\to\Omega _B$ such that $A=f(C)$, $B=g(C)$
\cite{bgl95,hz12,lah03}. In this way $A$ and $B$ can be simultaneously measured by measuring a single observable $C$. We say that
$A,B\in\oscript (H)$ are \textit{jointly measurable} if for all $\rho\in\sscript (H)$ there exist probability measures $\mu _\rho$ on
$\Omega _A\times\Omega _B$ such that $\mu _\rho\paren{\brac{x}\times\Omega _B}=\Phi _\rho ^A(x)$ and
$\mu _\rho\paren{\Omega _A\times\brac{y}}=\Phi _\rho ^B(y)$, for all $x\in\Omega _A$, $y\in\Omega _B$. We call $\mu _\rho$ the
\textit{joint distribution} of $A,B$ in the state $\rho$.

\begin{lem}    
\label{lem21}
{\rm{(i)}}\enspace For every $A\in\oscript (H)$, $\rho\in\sscript (H)$, $y\in f(\Omega _A)$ we have that
\begin{equation*} 
\Phi _\rho ^{f(A)}(y)=\Phi _\rho ^A\sqbrac{f^{-1}(y)}=\sum _x\brac{\Phi _\rho ^A(x)\colon f(x)=y}
\end{equation*}
{\rm{(ii)}}\enspace If $A$ and $B$ coexist, then $A$ and $B$ are jointly measurable.
\end{lem}
\begin{proof}
(i)\enspace For every $y\in\Omega _{f(A)}$ we obtain
\begin{align*}
\Phi _\rho ^{f(A)}&=\rmtr\sqbrac{\rho f(A)_y}=\rmtr\sqbrac{\rho A_{f^{-1}(y)}}=\Phi _\rho ^A\sqbrac{f^{-1}(y)}\\
  &=\sum _x\brac{\Phi _\rho ^A(x)\colon f(x)=y}\end{align*}
(ii)\enspace Since $A$ and $B$ coexist, there exists a $C\in\oscript (H)$ such that $A=f(C)$, $B=g(C)$. For $\rho\in\sscript (H)$, define the probability measure $\mu _\rho$ on $\Omega _A\times\Omega _B$ by
\begin{equation*} 
\mu _\rho (x,y)=\rmtr\sqbrac{\rho C_{f^{-1}(x)\cap g^{-1}(y)}}
\end{equation*}
We then obtain
\begin{equation*} 
\mu _\rho\paren{\brac{x}\times\Omega _B}=\rmtr\sqbrac{\rho C_{f^{-1}(x)}}=\rmtr\sqbrac{\rho f(C)_x}=\rmtr (\rho A_x)=\Phi _\rho ^A(x)
\end{equation*}
and in a similar way, $\mu _\rho\paren{\Omega _A\times\brac{y}}=\Phi _\rho ^B(y)$.
\end{proof}

We do not know whether the converse of Lemma~\ref{lem21}(ii) holds.

Let $\Omega _A$ be the outcome space for $A\in\oscript (H)$ and let $\Omega$ be another finite set. Suppose
$\mu\colon\Omega _A\times\Omega\to\sqbrac{0,1}$ satisfies $\sum\limits _{y\in\Omega}\mu _{xy}=1$ for every $x\in\Omega _A$. We call $\mu$ a \textit{transition probability} from $\Omega _A$ to $\Omega$. The condition $\sum _{y\in\Omega}\mu _{xy}=1$ says that $x$ transitions into some $y\in\Omega$ with probability ~1. A \textit{post-processing} of $A$ is an observable $B=\mu\tbullet A\in\oscript (H)$ with outcome space $\Omega$ defined by $B_y=\sum\limits _{x\in\Omega _A}\mu _{xy}A_x$ \cite{fhl18,gud320,gud420}. Notice that $B$ is indeed an observable because $B_y\ge 0$ for all $y\in\Omega$ and
\begin{equation*} 
\sum _{y\in\Omega}B_y=\sum _{x\in\Omega _A}\sum _{y\in\Omega}\mu _{xy}A_x=\sum _{x\in\Omega _A}A_x=I
\end{equation*}
We interpret $B=\mu\tbullet A$ as first measuring $A$ and then processing the result with transitions to the outcome space
$\Omega _B=\Omega$. One way of post-processing $A$ is by employing another observable $B$ and a collection of states $\alpha _x$,
$x\in\Omega _A$. We then define
\begin{equation*} 
\mu _{xy}=\rmtr (\alpha _xB_y)=\Phi _{\alpha _x}^B(y)
\end{equation*}
and write $\mu\tbullet A=\rmpost _{(\alpha ,B)}(A)$. We call $\rmpost _{(\alpha ,B)}(A)$ the \textit{post-processing} of $A$ relative to
$(\alpha _x,B)$. We can also post-process a probability measure $\nu$ on $\Omega _A$ to a probability measure $\mu\tbullet\nu$ on $\Omega$ by defining $(\mu\tbullet\nu )_y=\sum\limits _{x\in\Omega _A}\mu _{xy}\nu _x$.

\begin{lem}    
\label{lem22}
{\rm{(i)}}\enspace $\Phi _\rho ^{\mu\tbullet A}=\mu\tbullet \Phi _\rho ^A$.
{\rm{(ii)}}\enspace $\Phi _\rho ^{\rmpost _{(\alpha ,B)(A)}}(y)=\sum\limits _{x\in\Omega _A}\Phi _{\alpha _x}^B(y)\Phi _\rho ^A(x)$.
\end{lem}
\begin{proof}
(i)\enspace For all $y\in\Omega$ we have that
\begin{align*}
\Phi _\rho ^{\mu\tbullet A}(y)&=\rmtr\sqbrac{\rho (\mu\tbullet A)_y}=\rmtr\sqbrac{\rho\sum _{x\in\Omega _A}\mu _{xy}A_x}
   =\sum _{x\in\Omega _A}\mu _{xy}\rmtr (\rho A_x)\\
   &=\sum _{x\in\Omega _A}\mu _{xy}\Phi _\rho ^A(x)=\mu\tbullet\Phi _\rho ^A(y)
\end{align*}
The result follows. (ii)\enspace Applying (i) gives
\begin{equation*}
\Phi _\rho ^{\rmpost _{(\alpha ,B)}(A)}(y)=\sum _{x\in\Omega _A}\mu _{xy}\Phi _\rho ^A(x)
   =\sum _{x\in\Omega _A}\Phi _{\alpha _x}^B(y)\Phi _\rho ^A(x)\qedhere
\end{equation*}
\end{proof}

Let $A^i\in\oscript (H)$ with outcome spaces $\Omega ^i$, $i=1,2,\ldots ,n$ and let $\lambda _i\in (0,1)$ with $\sum\lambda _i=1$. A
\textit{generalized convex combination} of $A^i$ has outcome space $\Omega =\bigcup\limits _{i=1}^n\Omega ^i$ and is the observable define by
\begin{equation*}
\paren{\bigvee _{i=1}^n\lambda _iA^i}_x=\sum _i\brac{\lambda _iA_x^i\colon x\in\Omega ^i}
\end{equation*}
for all $x\in\Omega$. The two extreme cases of a generalized convex combination are when $\Omega _i=\Omega$, $i=1,2,\ldots ,n$ and when
$\Omega ^i\cap\Omega ^j=\emptyset$, $i\ne j$, $i,j=1,2,\ldots ,n$. The first case is called a \textit{convex combination} and is denoted by
$\sum\limits _{i=1}^n\lambda _iA^i$. The second case is called a \textit{convex union} and is denoted by $\bigcup\limits _{i=1}^n\lambda _iA^i$. We have that $\paren{\sum\limits _{i=1}^n\lambda _iA^i}_x=\sum\limits _{i=1}^n\lambda _iA_x^i$ for all $x\in\Omega$ and
$\paren{\bigcup _{i=1}^n\lambda _iA^i}_x=\lambda _jA_x^j$ where $x\in\Omega ^j$. When $A=\bigvee _{i=1}^n\lambda _iA^i$ we obtain
\begin{equation*}
\Phi _\rho ^A(x)=\sum _{i=1}^n\brac{\lambda _i\rmtr (\rho A_x^i)\colon x\in\Omega ^i}
   =\sum _{i=1}^n\brac{\lambda _i\Phi _\rho ^{A^i}(x)\colon x\in\Omega ^i}
\end{equation*}

\begin{exam}  
Let $\brac{a_1,a_2,a_3},\brac{b_1,b_2,b_3}\in\oscript (H)$. Define $A^1,A^2\in\oscript (H)$ by $\Omega _{A^i}=\brac{x_1,x_2,x_3}$,
$i=1,2$, $A_{x_j}^1=a_j$, $A_{x_j}^2=b_j$, $j=1,2,3$. Then for the convex combination $A=\tfrac{1}{2}\,A_1+\tfrac{1}{2}\,A^2$ we have that
$\Omega _A=\brac{x_1,x_2,x_3}$ and $A_{x_i}=\tfrac{1}{2}\,(a_i+b_i)$, $i=1,2,3$. Now define $B^1,B^2\in\oscript (H)$ by
$\Omega _{B^1}=\brac{x_1,x_2,x_3}$, $\Omega _{B^2}=\brac{y_1,y_2,y_3}$ where $\Omega _{B^1}\cap\Omega _{B^2}=\emptyset$ and
$B_{x_i}^1=a_i$, $B_{y_i}^2=b_i$, $i=1,2,3$. Then for the convex union $B=\tfrac{1}{2}\,B^1\cup\tfrac{1}{2}\,B^2$ we have
\begin{equation*}
\Omega _B=\Omega _{B^1}\cup\Omega _{B^2}=\brac{x_1,x_2,x_3,y_1,y_2,y_3}
\end{equation*}
and $B_{x_i}=\tfrac{1}{2}\,a_i$, $i=1,2,3$, $B_{y_i}=\tfrac{1}{2}\,b_i$, $i=1,2,3$. For another example, define $C^1,C^2\in\oscript (H)$ by
$\Omega _{C^1}=\brac{x_1,x_2,x_3}$, $\Omega _{C^2}=\brac{x_1,y_2,y_3}$ where $\brac{x_2,x_3}\cap\brac{y_2,y_3}=\emptyset$ and
$C_{x_i}^1=a_i$, $i=1,2,3$, $C_{x_1}^2=b_1$, $C_{y_i}^2=b_i$, $i=2,3$. Then for the generalized convex combination
$C=\tfrac{1}{2}\,C^1\vee\tfrac{1}{2}\,C^2$ we have
\begin{equation*}
\Omega _C=\Omega _{C^1}\cup\Omega _{C^2}=\brac{x_1,x_2,x_3,y_2,y_3}
\end{equation*}
and $C_{x_1}=\tfrac{1}{2}(a_1+b_2)$, $C_{x_2}=\tfrac{1}{2}\,a_2$, $C_{x_2}=\tfrac{1}{2}\,a_2$, $C_{y_2}=\tfrac{1}{2}\,b_2$,
$C_{y_3}=\tfrac{1}{2}\,b_3$. To illustrate the large number of possibilities even in this simple case, define $D^1,D^2\in\oscript (H)$ by
$\Omega _{D^1}=\brac{x_1,x_2,x_3}$, $\Omega _{D^2}=\brac{x_1,x_2,y_3}$ where $x_3\ne y_3$ and $D_{x_i}^1=a_i$, $i=1,2,3$,
$D_{x_1}^2=b_1$, $D_{x_2}^2=b_2$, $D_{y_3}^2=b_3$. For the generalized convex combination $D=\tfrac{1}{2}\,D^1\vee\tfrac{1}{2}\,D^2$ we have $\Omega _D=\brac{x_1,x_2,x_3,y_3}$ and $D_{x_1}=\tfrac{1}{2}(a_1+b_1)$, $D_{x^2}=\tfrac{1}{2}(a_2+b_2)$,$D_{x_3}=\tfrac{1}{2}\,a_3$, $D_{y_3}=\tfrac{1}{2}\,b_3$.\hfill\qedsymbol
\end{exam}

\begin{thm}    
\label{thm23}
{\rm{(i)}}\enspace $f\paren{\bigvee\limits _{i=1}^n\lambda _iA^i}_y=\sum\limits _{i,x}\brac{\lambda _iA_x^i\colon x\in\Omega _i,f(x)=y}$\newline
{\rm{(ii)}}\enspace $f\paren{\sum\limits _{i=1}^n\lambda _iA^i}=\sum\limits _{i=1}^n\lambda _if(A^i)$.
{\rm{(iii)}}\enspace $f\paren{\bigcup\limits _{i=1}^n\lambda _iA^i}=\sum _{i=1}^n\lambda _if|_{\Omega _i}(A^i)$.
\end{thm}
\begin{proof}
(i)\enspace Letting $f\colon\cup\Omega _i\to\Omega$ be a surjection, if $y\in\Omega$ we have that
\begin{align*}
f\paren{\bigvee _{i=1}^n\lambda _iA^i}_y&=\paren{\bigvee _{i=1}^n\lambda _iA^i}_{f^{-1}(y)}
  =\sum _x\brac{\paren{\bigvee _{i=1}^n\lambda _iA^i}_x\colon f(x)=y}\\
  &=\sum _{i,x}\brac{\lambda _iA_x^i\colon x\in\Omega _i,f(x)=y}
\end{align*}
(ii)\enspace In this case $\Omega _i=\Omega _j$ for all $i,j=1,\ldots ,n$ so $\cup\Omega _i=\Omega _j$ for all $j=1,2,\ldots ,n$. Hence, by (i) we obtain
\begin{equation*}
f\paren{\sum _{i=1}^n\lambda _iA^i}_y=\sum _{i=1}^n\brac{\lambda _iA_x^i\colon f(x)=y}=\sum _{i=1}^n\lambda _iA_{f^{-1}(y)}^i
  =\sum _{i=1}^n\lambda _if(A^i)_y
\end{equation*}
The result follows. (iii)\enspace In this case $\Omega _i\cap\Omega _j=\emptyset$ for all $i\ne j$. Hence, if $x\in\cup\Omega _i$, then
$x\in\Omega _i$ for a unique $i$. Then $\paren{\bigcup\limits _{i=1}^n\lambda _iA^i}_x=\lambda _iA_x^i$ where $x\in\Omega _j$ and by (i) we have that
\begin{align*}
f\paren{\bigcup _{i=1}^n\lambda _iA^i}_y&=\sum_i\sum _x\paren{\lambda _iA^i\colon f|_{\Omega _i}(x)=y}
   =\sum _i\lambda _iA_{(f|_{\Omega _i})^{-1}}^i(y)\\
   &=\sum _i\lambda _if|_{\Omega _i}(A^i)_y
\end{align*}
The result follows.
\end{proof}

Notice that $\mu\tbullet\paren{\sum\lambda _iA^i}=\sum\lambda _i(\mu\tbullet A^i)$ because
\begin{equation*}
\sqbrac{\mu\tbullet\paren{\sum\lambda _iA^i}}_y=\sum _x\mu _{xy}\sum _i\lambda _iA_x^i=\sum _i\lambda _i\sum _x\mu _{xy}A^i
   =\sum _i\lambda _i(\mu\tbullet A^i)_y
\end{equation*}
In general, $\mu\tbullet\paren{\bigvee\lambda _iA^i}\ne\bigvee\lambda _i(\mu\tbullet A^i)$ because the $A^i$ can have different outcome spaces so $\mu\tbullet A^i$ is not defined.

We call the observables $I^{x_j}$ with outcome space $\Omega =\brac{x_1,x_2,\ldots ,x_n}$ defined by $I_{x_i}^{x_j}=\delta _{ij}I$
\textit{identity observables}. If $A\in\oscript (H)$ with $\Omega _A=\Omega$ and $\lambda\in\parensq{0,1}$, we call $B\in\oscript (H)$ given by $B=(1-\lambda )I^{x_j}+\lambda A$ the observable $A$ \textit{with noise factor} $(1-\lambda )/\lambda$ \cite{hz12,hrsz09}.

\begin{thm}    
\label{thm24}
If $A^i\in\oscript (H)$ and $\lambda _i\in\sqbrac{0,1}$ with $\sum\lambda _i=1$, $i=1,2,\ldots ,n$, then there exists $A\in\oscript (H)$ and surjections $f_j\colon\Omega _A\to\Omega _{A^j}$ such that
\begin{equation*}
f_j(A)=(1-\lambda _j)I^{x_j}+\lambda _jA^j
\end{equation*}
for some $x_j\in\Omega _{A^j}$, $j=1,2,\ldots ,n$.
\end{thm}
\begin{proof}
We can assume, without loss of generality, that $\Omega _{A^i}\cap\Omega _{A^j}=\emptyset$ for $i\ne j$. Letting
$A=\bigcup\limits _{i=1}^n\lambda _iA^i$ we have that $\Omega _A=\bigcup\limits _{i=1}^n\Omega _{A^i}$. Define the functions $f_j\colon\Omega _A\to\Omega _{A^j}$ by $f_j(x)=x$ for all $x\in\Omega _{A^j}$ and $f_j(x)=x_j\in\Omega _{A^j}$ for $x\notin\Omega _{A^j}$. Then we obtain
\begin{align*}
f_j\paren{\bigcup _{i=1}^n\lambda _iA^i}_{x_j}&=\paren{\bigcup\lambda _iA^i}_{f_j^{-1}(x_j)}
  =\sum _i\lambda _iA_{f_j^{-1}(x_j)\cap\Omega _i}^i\\
   &=\sum _{i\ne j}\lambda _iA_{\Omega _i}^i+\lambda _jA_{x_j}^j=\sum _{i\ne j}\lambda _iI^{x_j}+\lambda _jA_{x_j}^j
\end{align*}
and for $x=x_j$ we have that
\begin{equation*}
f_j\paren{\bigcup _{i=1}^n\lambda _iA^i}_x=\paren{\bigcup\lambda _iA^i}_{f_j^{-1}(x)}=\paren{\bigcup\lambda _iA^i}_x=\lambda _jA_x^j
\end{equation*}
Hence, if $A=\bigcup\limits _{i=1}^n\lambda _iA^i$ then
\begin{equation*}
f_j(A)=(1-\lambda )I^{x_j}+\lambda _jA^j\qedhere
\end{equation*}
\end{proof}

Since the $f_j(A)$ in Theorem~\ref{thm24} are all parts of the same observable $A$, we see that the $f_j(A)=(1-\lambda _j)I^{x_j}+\lambda _jA^j$ mutually coexist, $j=1,2,\ldots ,n$. We conclude that any set of observables $A^j$, $j=1,2,\ldots ,n$, ``almost coexist'' in the sense that a noisy version of $A^j$ is a part of an observable $A$, $j=1,2,\ldots ,n$.

For $A,B\in\oscript (H)$, we define their \textit{sequential product} $A\circ B\in\oscript (H)$ \cite{gud120,gud220} by
$\Omega _{A\circ B}=\Omega _A\times\Omega _B$ and
\begin{equation*}
(A\circ B)_{(x,y)}=A_x\circ B_y=A_x^{1/2}B_yA_x^{1/2}
\end{equation*}
If $X\subseteq\Omega _A\times\Omega _B$, we have that
\begin{equation*}
(A\circ B)_X=\sum _{(x,y)\in X}(A\circ B)_{(x,y)}=\sum _{(x,y)\in X}A_x\circ B_y
\end{equation*}
It follows that $(A\circ B)_{\brac{x}\times Y}=A_x\circ B_Y$ but $(A\circ B)_{X\times\brac{y}}\ne A_X\circ B_y$, in general. Moreover,
\begin{equation*}
\Phi _\rho ^{A\circ B}(x,y)=\rmtr\sqbrac{\rho (A\circ B)_{(x,y)}}=\rmtr (\rho A_x\circ B_y)
\end{equation*}
and if $X\subseteq\Omega _A\times\Omega _B$ then
\begin{equation*}
\Phi _\rho ^{A\circ B}(X)=\sum _{(x,y)\in X}\rmtr (\rho A_x\circ B_y)
\end{equation*}

We also define the observable $(B\mid A)$ with $\Omega _{(B\mid A)}=\Omega _B$ and
$(B\mid A)_x=\sum\limits _{x\in\Omega _A}(A_x\circ B_y)$. We call $(B\mid A)$ the observable $B$ \textit{conditioned on} $A$ \cite{gud120,gud220,gud420}. We then have that
\begin{equation*}
\Phi _\rho ^{(B\mid A)}(y)=\rmtr\sqbrac{\rho (B\mid A)_y}=\rmtr\sqbrac{\rho\sum _{x\in\Omega _A}(A_x\circ B_y)}
  =\sum _{x\in\Omega _A}\rmtr (\rho A_x\circ B_y)
\end{equation*}
for all $Y\subseteq\Omega _B$ we obtain
\begin{equation*}
\Phi _\rho ^{(B\mid A)}(Y)=\sum _{y\in Y}\sum _{x\in\Omega _A}\rmtr (\rho A_x\circ B_y)=\sum _{x\in\Omega _A}\rmtr (\rho A_x\circ B_Y)
\end{equation*}
Defining the functions $f\colon\Omega _A\times\Omega _B\to\Omega _B$, $g\colon\Omega _A\times\Omega _B\to\Omega _A$ by
$f(x,y)=y$ for all $x\in\Omega _A$ and $g(x,y)=x$ for all $y\in\Omega _B$ we see that
\begin{align*}
f(A\circ B)_y&=(A\circ B)_{f^{-1}(y)}=\sum _x\brac{(A\circ B)_{(x,y)}\colon f(x,y)=y}\\
  &=\sum _{x\in\Omega _A}(A_x\circ B_y)=(B\mid A)_y\\
\intertext{and}
g(A\circ B)_x&=(A\circ B)_{g^{-1}(x)}=\sum _y\brac{(A\circ B)_{(x,y)}\colon g(x,y)=x}\\
  &=\sum _{y\in\Omega _B}(A_x\circ B_y)=A_x
\end{align*}
Hence, $(B\mid A)=f(A\circ B)$ and $A=g(A\circ B)$. We conclude that $(B\mid A)$ and $A$ coexist. In general, $(B\mid A)$ and $B$ need not coexist. Also $(B\mid A)$ and $(C\mid A)$ need not coexist even though they both coexist with $A$.

\begin{thm}    
\label{thm25}
{\rm{(i)}}\enspace $A\circ\paren{\bigvee\limits _{i=1}^n\lambda _iB^i}=\bigvee\limits _{i=1}^n\lambda _iA\circ B^i$.\newline
{\rm{(ii)}}\enspace $\paren{\bigvee _{i=1}^n\lambda _iB^i\mid A}=\bigvee _{i=1}^n\lambda _i(B^i\mid A)$.
\end{thm}
\begin{proof}
(i)\enspace For all $x\in\Omega _A$, $y\in\bigcup\Omega _i$ with $\Omega _i=\Omega _{B_i}$ we have that
\begin{align*}
\paren{A\circ\bigvee _{i=1}^n\lambda _iB^i}_{(x,y)}&=A_x\circ\paren{\bigvee\lambda _iB^i}_y
  =A_x\circ\sum _i\brac{\lambda _iB_y^i\colon y\in\Omega _i}\\
  &=\sum _i\brac{\lambda _iA_x\circ B_y^i\colon y\in\Omega _i}\\
  &=\brac{\bigvee\lambda _i(A\circ B^i)_{(x,y)}\colon y\in\Omega _i}\\
  &=\paren{\bigvee _{i=1}^n\lambda _iA\circ B^i}_{(x,y)}
\end{align*}
The result follows. (ii)\enspace For all $x\in\Omega _A$, $y\in\bigcup\Omega _i$, it follows from (i) that
\begin{align*}
\paren{\bigvee _{i=1}^n\lambda _iB^i\mid A}_{(x,y)}&=\sum _{x\in\Omega _A}A_x\circ\paren{\bigvee\lambda _iB^i}_y
   =\sum _{x\in\Omega _A}\bigvee\lambda _iA_x\circ B_y^i\\
   &=\sum _{x\in\Omega _A}\sum _i\brac{\lambda _iA_x\circ B_y^i\colon y\in\Omega _i}
   =\sum _{y\in\Omega _i}\lambda _i\sum _{x\in\Omega _A}A_x\circ B_y^i\\
   &=\bigvee\paren{\lambda _i\sum _{x\in\Omega _A}A_x\circ B_y^i}=\sqbrac{\bigvee\lambda _i(B^i\mid A)}_{(x,y)}
\end{align*}
The result follows.
\end{proof}

In general, $\paren{\bigvee\lambda _iB^i}\circ A\ne\bigvee\lambda _i(B^i\circ A)$ and
$\paren{A\mid\bigvee\lambda _iB^i}\ne\bigvee\lambda _i(A\mid B^i)$.

If $A\in\oscript (H_1)$ $B\in\oscript (H_2)$, we define the \textit{tensor product} $A\otimes B\in\oscript (H_1\otimes H_2)$ \cite{gud120,gud420} by $\Omega _{A\otimes B}=\Omega _A\times\Omega _B$ and $(A\otimes B)_{(x,y)}=A_x\times B_y$. If $\mu _{xy}$ and $\nu _{uv}$ are transition probabilities, we define the transition probability
\begin{equation*}
\mu\tbullet\nu_{\paren{(x,u),(y,v)}}=\mu _{xy}\nu _{uv}
\end{equation*}
We see that $\mu\tbullet\nu$ is indeed a transition probability because
\begin{equation*}
\sum _{(y,v)}\mu\tbullet\nu _{\paren{(x,u),(y,v)}}=\sum _{y,v}\mu _{xy}\nu _{uv}=\sum _y\mu _{xy}\sum _v\nu _{uv}=1
\end{equation*}
If $f\colon\Omega _A\to\Omega _1$, $g\colon\Omega _B\to\Omega _2$, we define the function
$f\times g\colon\Omega _A\times\Omega _B\to\Omega _1\times\Omega _2$ by
\begin{equation*}
f\times g(x,y)=\paren{f(x),g(y)}
\end{equation*}
The next result summarizes combinations with $A\otimes B$.

\begin{thm}    
\label{thm26}
{\rm{(i)}}\enspace If $A\in\oscript (H_1)$, $B\in\oscript (H_2)$, then $(\mu\tbullet A)\otimes (\nu\tbullet B)=(\mu\tbullet\nu )\tbullet (A\otimes B)$.
{\rm{(ii)}}\enspace If $A\in\oscript (H_1)$, $B\in\oscript (H_2)$, then $f(A)\otimes g(B)=f\times g(A\otimes B)$.
{\rm{(iii)}}\enspace $A\otimes\paren{\bigvee\lambda _iB^i}=\bigvee\lambda _iA\otimes B^i$ and
$\paren{\bigvee\lambda _iB^i}\otimes A=\bigvee\lambda _iB^i\otimes A$.
{\rm{(iv)}}\enspace If $A,C\in\oscript (H_1)$ and $B,D\in\oscript (H_2)$, then
\begin{equation*}
\sqbrac{(A\otimes B)\circ (C\otimes D)}_{\paren{(x,y),(u,v)}}=\sqbrac{(A\circ C)\otimes (B\circ D)}_{\paren{(x,u),(y,v)}}
\end{equation*}
\end{thm}
\begin{proof}
(i)\enspace For all applicable $x,y,u,v$ we have that
\begin{align*}
\sqbrac{(\mu\tbullet A)\otimes (\nu\tbullet B)}_{(y,z)}&=(\mu\tbullet A)_y\otimes (\nu\tbullet B)_z
  =\sum _x\mu _{xy}A_x\otimes\sum _u\nu _{uz}B_u\\
  &=\sum _{x,u}\mu _{xy}\nu _{uz}(A\otimes B)_{(x,u)}\\
 & =\sum _{x,u}\mu\tbullet\nu _{\paren{(x,u),(y,z)}}(A\circ B)_{(x,u)}\\
  &=\sqbrac{(\mu\tbullet\nu )\tbullet (A\otimes B)}_{(y,z)}
\end{align*}
The result follows. (ii)\enspace Letting $h=f\times g$ we obtain
\begin{align*}
\sqbrac{f(A)\otimes g(B)}_{(u,v)}&=f(A)_u\otimes g(B)_v=A_{f^{-1}(u)}\otimes B_{g^{-1}(v)}\\
   &=\sum _x\brac{A_x\colon f(x)=u}\otimes\sum _y\brac{B_y\colon g(y)=v}\\
   &=\sum _{x,y}\brac{A_x\otimes B_y\colon f(x)=u, g(y)=v}\\
   &=\sum _{x,y}\brac{A_x\otimes B_y\colon h(x,y)=(u,v)}\\
   &=\sum _{x,y}\brac{(A\otimes B)_{(x,y)}\colon h(x,y)=(u,v)}\\
   &=(A\otimes B)_{h^{-1}(u,v)}=\sqbrac{h(A\otimes B)}_{(u,v)}
\end{align*}
The result follows. (iii)\enspace For all applicable $x,y$ we have that
\begin{align*}
\sqbrac{A\otimes\paren{\bigvee\lambda _iB^i}}_{(x,y)}&=A_x\otimes\paren{\bigvee\lambda _iB^i}_y
   =A_x\otimes\sum _i\brac{\lambda _iB_y^i\colon y\in\Omega _i}\\
   &=\sum _i\brac{\lambda _iA_y\otimes B_y^i\colon y\in\Omega _i}\\
   &=\sum _i\brac{\lambda _iA_x\otimes B_y^i\colon y\in\Omega _i}\\
   &=\paren{\bigvee\lambda _iA\otimes B^i}_{(x,y)}
\end{align*}
The result follows. (iv)\enspace For all applicable $x,y,u,v$ we obtain
\begin{align*}
&\sqbrac{(A\otimes B)\circ (C\otimes D)}_{\paren{(x,y),(u,v)}}\\
   &\qquad =(A\otimes B)_{(x,y)}\circ (C\otimes D)_{(u,v)}=(A_x\otimes B_y)\circ (C_u\otimes D_v)\\
   &\qquad =(A_x\otimes B_y)^{1/2}(C_u\otimes D_v)(A_x\otimes B_y)^{1/2}\\
   &\qquad =(A_x^{1/2}\otimes B_y^{1/2})(C_u\otimes D_v)(A_x^{1/2}\otimes B_y^{1/2})\\
   &\qquad =A_x^{1/2}C_uA_x^{1/2}\otimes B_y^{1/2}D_vB_y^{1/2}=A_x\circ C_u\otimes B_y\circ D_v\\
   &\qquad =(A\circ C)_{(x,u)}\otimes (B\circ D)_{(y,v)}=\sqbrac{(A\circ C)\otimes (B\circ D)}_{\paren{(x,u),(y,v)}}\qedhere
\end{align*}
\end{proof}

We see from Theorem~\ref{thm26}(iv) that $(A\otimes B)\circ (C\otimes D)\ne (A\circ C)\otimes (B\otimes D)$, in general. It also follows from Theorem~\ref{thm26}(ii) that if $A,B\in\oscript (H_1)$ coexist and $C,D\in\oscript (H_2)$ coexist, then $A\otimes C$ and $B\otimes D$ coexist. Indeed, we have observables $E\in\oscript (H_1)$, $F\in\oscript (H_2)$ and functions $f_1,g_2,f_2,g_2$ such that $A=f_1(E)$, $B=g_1(E)$,
$C=f_2(F)$, $D=g_2(F)$. Applying Theorem~\ref{thm26}(ii) gives
\begin{align*}
A\otimes C&=f_1(E)\otimes f_2(F)=f_1\times f_2(E\otimes F)\\
B\otimes D&=g_1(E)\otimes g_2(F)=g_1\times g_2(E\otimes F)
\end{align*}
Hence, $A\otimes C$ and $B\otimes D$ coexist.

We have gone from $\oscript (H_1)$, $\oscript (H_2)$ to obtain observables in $\oscript (H_1\otimes H_2)$. We can also go the other way to reduce observables in $\oscript (H_1\otimes H_2)$ to elements of $\oscript (H_1)$ and $\oscript (H_2)$. If $A\in\oscript (H_1\otimes H_2)$, we define the \textit{reduced observables} $A^1\in\oscript (H_1)$, $A^2\in\oscript (H_2)$ by $A_x^1=\tfrac{1}{n_2}\,\rmtr _{H_2}A_x$ for all
$x\in\Omega _A$ where $n_2=\dim H_2$ and $\rmtr _{H_2}$ is the partial trace with respect to $H_2$ \cite{gud120,gud420,hz12} and similarly $A_x^2=\tfrac{1}{n_1}\,\rmtr _{H_1}A_x$. To check that $A^1$ is indeed an observable, we see that $A_x^1\ge 0$ and
\begin{align*}
\sum _{x\in\Omega _A}A_x^1&=\frac{1}{n_2}\sum _{x\in\Omega _A}\rmtr _{H_2}(A_x)
  =\frac{1}{n_2}\,\rmtr _{H_2}\paren{\sum _{x\in\Omega _A}A_x}=\tfrac{1}{n_2}\,\rmtr _{H_2}(I)\\
  &=\tfrac{1}{n_2}\,\rmtr _{H_2}(I_1\otimes I_2)=\tfrac{1}{n_2}\,(I_2)I_1=I_1
\end{align*}
where $I_1$, $I_2$ are the identity operators on $H_1$, $H_2$, respectively.

If $A\in\oscript (H_1)$, $B\in\oscript (H_2)$, we have the observable $C=A\otimes B\in\oscript (H_1\otimes H_2)$. It is interesting to note that
\begin{equation*}
(A\otimes B)_{(x,y)}^1=\tfrac{1}{n_2}\,\rmtr _{H_2}(A\otimes B)_{(x,y)}=\tfrac{1}{n_2}\,\rmtr _{H_2}(A_x\otimes B_y)
  =\tfrac{1}{n_2}\,(\rmtr B_y)A_x
\end{equation*}
Hence, $(A\otimes B)_{\brac{x}\times\Omega _B}^1=A_x$ and $(A\otimes B)_{\Omega _A\times\brac{y}}^1=\tfrac{1}{n_2}(\rmtr B_y)I_1$. In a similar way, $(A\otimes B)_{(x,y)}^2=\tfrac{1}{n_1}(\rmtr A_x)B_y$. For $A,B\in\oscript (H_1\otimes H_2)$ we obtain
\begin{equation*}
(A\circ B)_{(x,y)}^1=\tfrac{1}{n_2}\,\rmtr _{H_2}(A\circ B)_{(x,y)}=\tfrac{1}{n_2}\,\rmtr _{H_2}A_x\circ B_y
   =\tfrac{1}{n_2}\,\rmtr _{H_2}A_x^{1/2}B_yA_x^{1/2}
\end{equation*}
On the other hand, we have that
\begin{align*}
(A^1\circ B^1)_{(x,y)}&=A_x^1\circ B_y^1=\tfrac{1}{(n_2)^2}(\rmtr _{H_2}A_x^1)\circ (\rmtr _{H_2}B_y^1)\\
   &=\tfrac{1}{(n_2)^2}(\rmtr _{H_2}A^1)^{1/2}(\rmtr _{H_2}B_y^1)(\rmtr _{H_2}A^1)^{1/2}
\end{align*}
It follows that $(A\circ B)^1\ne A^1\circ B^1$, in general, so sequential products need not be preserved under reduction. The next result shows that the other combinations we considered are preserved.

\begin{thm}    
\label{thm27}
{\rm{(i)}}\enspace If $A\in\oscript (H_1\otimes H_2)$, then $f(A)^i=f(A^i)$, $i=1,2$.
{\rm{(ii)}}\enspace If $A,B\in\oscript (H_1\otimes H_2)$ coexist, then $A^i$ and $B^i$ coexist, $i=1,2$.
{\rm{(iii)}}\enspace If $A\in\oscript (H_1\otimes H_2)$, then $(\mu\tbullet A)^i=\mu\tbullet A^i$, $i=1,2$.
{\rm{(iv)}}\enspace If $A^j\in\oscript (H_1\otimes H_2)$, then $\paren{\bigvee\lambda _jA^j}^i=\bigvee\lambda _j(A^j)^i$, $i=1,2$.
\end{thm}
\begin{proof}
We prove these results for $i=1$ and the proofs for $i=2$ are similar. (i)\enspace For all $y\in\Omega _{f(A)}$ we obtain
\begin{align*}
f(A)_y^1&=\tfrac{1}{n_2}\,\rmtr _{H_2}f(A)_y=\tfrac{1}{n_2}\,\rmtr _{H_2}A_{f^{-1}(y)}
   =\tfrac{1}{n_2}\rmtr _{H_2}\paren{\sum _x\brac{A_x\colon f(x)=y}}\\
   &=\sum _x\brac{\tfrac{1}{n_2}\,\rmtr _{H_2}A_x\colon f(x)=y}=\sum _x\brac{A_x^1\colon f(x)=y}=f(A^1)_y
\end{align*}
The result follows. (ii)\enspace If $A,B$ coexist, there exist $C\in\oscript (H_1\otimes H_2)$ and functions $f,g$ such that $A=f(C)$, $B=g(C)$. Applying (i) gives $A^1=f(C)^1=f(C^1)$ and $B^1=g(C)^1=g(C)^1$. Hence, $A^1$ and $B^1$ coexist.
(iii)\enspace For all applicable $y$ we have that
\begin{align*}
(\mu\tbullet A)_y^1&=\tfrac{1}{n_2}\,\rmtr _{H_2}(\mu\tbullet A)_y
   =\tfrac{1}{n_2}\,\rmtr _{H_2}\paren{\sum _x\mu _{xy}A_x}=\sum _x\mu _{xy}\,\tfrac{1}{n_2}\,\rmtr _{H_2}A_x\\
  &=\sum _x\mu _{xy}A_x^1=(\mu\tbullet A^1)_y
\end{align*}
The result follows. (iv)\enspace For all applicable $x$ we obtain
\begin{align*}
\paren{\bigvee\lambda _iA^i}_x^1&=\tfrac{1}{n_2}\,\rmtr _{H_2}\paren{\bigvee\lambda _iA^i}_x
  =\tfrac{1}{n_2}\,\rmtr _{H_2}\sqbrac{\sum _i\brac{\lambda _iA_x^i\colon x\in\Omega _i}}\\
  &=\sum _i\brac{\lambda _i\rmtr _{H_2}A_x^i\colon x\in\Omega _i}=\sum _i\brac{\lambda _i(A_x^i)^1\colon x\in\Omega _i}\\
  &=\sqbrac{\bigvee\lambda _i(A^i)^1}_x
\end{align*}
This proves the result.
\end{proof}

Although we have not found a counterexample, we conjecture that the converse of Theorem~\ref{thm27}(ii) does not hold.

\section{Instruments}  
For instruments, we define post-processing $\mu\tbullet\iscript$, parts $f(\iscript )$, coexistence and generalized convex combinations
$\bigvee\lambda _i\iscript ^i$ as we did for observables. The next theorem shows that these definitions are consistent.

\begin{thm}    
\label{thm31}
{\rm{(i)}}\enspace $f(\iscript )^\wedge =f(\iscripthat\,)$
{\rm{(ii)}}\enspace $\paren{\bigvee\lambda _i\iscript ^i}^\wedge =\bigvee\lambda _i\iscript ^{i\wedge}$.
{\rm{(iii)}}\enspace $(\mu\tbullet\iscript )^\wedge =\mu\tbullet\iscripthat$.
\end{thm}
\begin{proof}
(i)\enspace For all $x\in\Omega _\iscript$ and $\rho\in\sscript (H)$ we have that
\begin{align*}
\rmtr\sqbrac{\rho f(\iscripthat\,)_x}&=\rmtr\sqbrac{\rho\iscripthat _{f^{-1}(x)}}=\rmtr\sqbrac{\iscript _{f^{-1}(x)}(\rho )}\\
   &=\rmtr\sqbrac{f(\iscript )_x(\rho )}=\rmtr\sqbrac{\rho f(\iscript )_x^\wedge}
\end{align*}
Hence, $f(\iscript )_x^\wedge=f(\iscripthat\,)_x$ for all $x\in\Omega _\iscript$ and the result follows.
(ii)\enspace For all $x\in\bigcup\Omega _i$ where $\Omega _i=\Omega _{\iscript _i}$ we obtain
\begin{align*}
\rmtr\sqbrac{\rho\paren{\bigvee\lambda _i\iscript ^{i\wedge}}_x}
   &=\rmtr\sqbrac{\rho\sum _x\brac{\lambda _i\iscript _x^{i\wedge}\colon x\in\Omega _i}}
   =\sum _x\brac{\lambda _i\rmtr (\rho\iscript _x^{i\wedge})\colon x\in\Omega _i}\\
   &=\sum _x\brac{\lambda _i\rmtr\sqbrac{\iscript _x^i(\rho )}\colon x\in\Omega _i}\\
   &=\rmtr\sqbrac{\sum _x\brac{\lambda _i\iscript _x^i(\rho )\colon x\in\Omega _i}}\\
   &=\rmtr\sqbrac{\paren{\bigvee\lambda _i\iscript ^i}_x(\rho )}=\rmtr\sqbrac{\rho\paren{\bigvee\lambda _i\iscript}_x^\wedge}
\end{align*}
Hence, $\bigvee\lambda _i\iscript _x^{i\wedge}=\paren{\bigvee\lambda _i\iscript ^i}_x^\wedge$ for all $x\in\bigcup\Omega _i$ and this gives the result.
(iii)\enspace For all $y\in\Omega _{\mu\tbullet\iscript}$ we have that
\begin{align*}
\rmtr\sqbrac{\rho (\mu\tbullet\iscripthat\,)_y}&=\rmtr\paren{\rho\sum _x\mu _{xy}\iscripthat _x}=\sum _x\mu _{xy}\rmtr (\rho\iscripthat _x))
   =\sum _x\mu _{xy}\rmtr\sqbrac{\iscript _x(\rho )}\\
   &=\rmtr\sqbrac{\sum _x\mu _{xy}\iscript _x(\rho )}=\rmtr\sqbrac{(\mu\tbullet\iscript )_y(\rho )}=\rmtr\sqbrac{\rho (\mu\tbullet\iscript )_y^\wedge}
\end{align*}
We conclude that $(\mu\tbullet\iscript )_y^\wedge =(\mu\tbullet\iscripthat\,)_y$ for all $y\in\Omega _{\mu\tbullet\iscript}$ and this proves the result.
\end{proof}

Applying Theorem~\ref{thm31}(i) we obtain the following.

\begin{cor}    
\label{cor32}
If $\iscript ,\jscript\in\rmin (H)$ coexist, then $\iscripthat$, $\jscripthat$ coexist.
\end{cor}

Unlike the other concepts, we must define sequential products of instruments differently from that of observables. If
$\iscript ,\jscript\in\rmin (H)$, then their \textit{sequential product} $\iscript\circ\jscript\in\rmin (H)$ is defined by
$\Omega _{\iscript\circ\jscript}=\Omega _\iscript\times\Omega _\jscript$ and
$(\iscript\circ\jscript )_{x,y}(\rho )=\jscript _y\sqbrac{\iscript _x(\rho )}$ for all $\rho\in\sscript (H)$. We define the
\textit{conditional instrument} $(\jscript\mid\iscript )\in\rmin (H)$ by $\Omega _{(\jscript\mid\iscript )}=\Omega _\jscript$ and
\begin{equation*}
(\jscript\mid\iscript )_y(\rho )=\sum _{x\in\Omega _\iscript}(\iscript\circ\jscript )_{(x,y)}(\rho )
   =\sum _{x\in\Omega _\iscript}\jscript _y\paren{\iscript _x(\rho )}=\jscript _y\sqbrac{\iscript _{\Omega _\iscript (\rho )}}
\end{equation*}
Of course, $\iscript _{\Omega _\iscript}$ is the channel given by $\iscript$. Unlike for observables, the next theorem has a second part.

\begin{thm}    
\label{thm33}
{\rm{(i)}}\enspace $\iscript\circ\paren{\bigvee\lambda _i\jscript ^i}\!=\!\bigvee (\lambda _i\iscript\circ\jscript ^i)$.
{\rm{(ii)}}\enspace $\paren{\bigvee\lambda _i\jscript ^i}\circ\iscript\!=\!\bigvee (\lambda _i\jscript ^i\circ\iscript )$
\end{thm}
\begin{proof}
We let $x\in\Omega _\iscript$, $y\in\bigcup\Omega _i$ where $\Omega _i=\Omega _{\jscript ^i}$ and $\rho\in\sscript (H)$ be arbitrary elements.
(i)\enspace The following steps hold:
\begin{align*}
\paren{\iscript\circ\bigvee\lambda _i\jscript ^i}_{(x,y)}(\rho )&=\paren{\bigvee\lambda _i\jscript ^i}_y\paren{\iscript _x(\rho )}
   =\sum _y\brac{\lambda _i\jscript _y^i\sqbrac{\iscript _x(\rho )}\colon y\in\Omega _i}\\
   &=\sum _y\brac{\lambda _i(\iscript\circ\jscript )_{(x,y)}^i(\rho )\colon y\in\Omega _i}\\
   &=\paren{\bigvee\lambda _i\iscript\circ\jscript ^i}_{(x,y)}(\rho )
\end{align*}
The result now follows. (ii)\enspace The following steps hold:
\begin{align*}
\paren{\bigvee\lambda _i\jscript ^i\circ\iscript}_{(x,y)}(\rho )&=\iscript _y\paren{\bigvee\lambda _i\jscript _x^i(\rho )}
   =\iscript _y\sqbrac{\sum _x\brac{\lambda _i\jscript _x^i(\rho )\colon x\in\Omega _i}}\\
   &=\sum _x\brac{\lambda _i\iscript _y\paren{\jscript _x^i(\rho )}\colon x\in\Omega _i}\\
   &=\sum _x\brac{\lambda _i(\jscript ^i\circ\iscript )_{(x,y)}(\rho )\colon x\in\Omega _i}\\
   &=\bigvee\paren{\lambda _i\jscript ^i\circ\iscript}_{(x,y)}(\rho )
\end{align*}
The result follows.
\end{proof}

Most of the theorems in Section~2 concerning observables hold for instruments and the proofs are similar so we shall not repeat them. We will mainly concentrate on various types of instruments that we now define. We say that an instrument $\iscript\in\rmin (H)$ is:
\begin{description}
\item[\textit{Kraus}]
if it has the form $\iscript _x(\rho )=S_x\rho S_x^*$ where $S_x\in\lscript (H)$ with $\sum S_x^*S_x=I$,
\item[\textit{L\"uders}]
if $\iscript _x(\rho )=\lscript _x^A(\rho )=A_x^{1/2}\rho A_x^{1/2}$ where $A\in\oscript (H)$,
\item[\textit{Trivial}]
if $\iscript _x(\rho )=\rmtr (\rho A_x)\alpha$ where $A\in\oscript (H)$, $\alpha\in\sscript (H)$,
\item[\textit{Semitrivial}]
if $\iscript _x(\rho )=\rmtr (\rho A_x)\alpha _x$ where $A\in\oscript (H)$, $\alpha _x\in\sscript (H)$.
\end{description}
Notice that a L\"uders instrument is a special case of a Kraus instrument and a trivial instrument is a special case of a semitrivial instrument. An interesting example of a semitrivial instrument is
\begin{equation*}
\iscript _x(\rho )=\frac{\rmtr (\rho A_x)}{\rmtr (A_x)}\,A_x
\end{equation*}
It is easy to check that the observable measured by the Kraus instrument is $\iscripthat _x=S_x^*S_x$ and the other three types of instruments measure the observable $A$. This also shows that an observable is measured by many different instruments. We call $S_x$ the
\textit{operators} for the Kraus instrument $\iscript _x(\rho )=S_x\rho S_x^*$. We say that two observables $A,B\in\oscript (H)$
\textit{commute} if $A_xB_y=B_yA_x$ for all $x\in\Omega _A$, $y\in\Omega _B$.

\begin{thm}    
\label{thm34}
{\rm{(i)}}\enspace $(\lscript ^A\circ\lscript ^B)^\wedge =(\lscript ^A)^\wedge\circ (\lscript ^B)^\wedge =A\circ B$.
{\rm{(ii)}}\enspace $\lscript ^A\circ\lscript ^B$ is a L\"uders instrument if and only if $A$ and $B$ commute and ißßn this case
$\lscript ^A\circ\lscript ^B=\lscript ^{A\circ B}$.
{\rm{(iii)}}\enspace If $\iscript$ and $\jscript$ are Kraus instruments with operators $S_x$, $T_y$, respectively, then $\iscript\circ\jscript$ is a Kraus instrument with operators $T_yS_x$.
{\rm{(iv)}}\enspace If $\iscript$, $\jscript$ are simitrivial with observables $A$, $B$ and states $\alpha _x$, $\beta _y$, respectively, then
$\iscript\circ\jscript$ is semitrivial with observable $C_{(x,y)}=\rmtr (\alpha _xB_y)A_x$ and states $\beta _y$. Moreover, $(\jscript\mid\iscript )$ is semitrivial with observable $\rmpost _{(\alpha ,B)}(\iscripthat\,)$ and states $\beta _y$.
\end{thm}
\begin{proof}
(i)\enspace For all $x\in\Omega _A$, $y\in\Omega _B$ and $\rho\in\sscript (H)$ we have that
\begin{align*}
\rmtr\sqbrac{\rho (\lscript ^A\circ\lscript ^B)_{(x,y)}^\wedge}&=\rmtr\sqbrac{(\lscript ^A\circ\lscript ^B)_{(x,y)}(\rho )}
   =\rmtr\sqbrac{\lscript _y^B\paren{\lscript _x^A(\rho )}}\\
   &=\rmtr (B_y^{1/2}A_x^{1/2}\rho A_x^{1/2}B_y^{1/2})=\rmtr (\rho A_x^{1/2}B_yA_x^{1/2})\\
   &=\rmtr\sqbrac{\rho (A\circ B)_{(x,y)}}
\end{align*}
It follows that $(\lscript ^A\circ\lscript ^B)^\wedge =A\circ B=(\lscript ^A)^\wedge\circ (\lscript ^B)^\wedge$.
(ii)\enspace As in (i) we have that
\begin{equation}                
\label{eq31}
(\lscript ^A\circ\lscript ^B)_{(x,y)}(\rho )=B_y^{1/2}A_x^{1/2}\rho A_x^{1/2}B_y^{1/2}
\end{equation}
On the other hand
\begin{equation*}
\lscript _{(x,y)}^{A\circ B}(\rho )=(A\circ B)_{(x,y)}^{1/2}\rho (A\circ B)_{(x,y)}^{1/2}
   =(A_x^{1/2}B_yA_x^{1/2})^{1/2}\rho (A_x^{1/2}B_yA_x^{1/2})^{1/2}
\end{equation*}
By \eqref{eq31}, we conclude that $\lscript ^A\circ\lscript ^B$ is a L\"uders instrument if and only if $B_y^{1/2}A_x^{1/2}=A_x^{1/2}B_y^{1/2}$ which is equivalent to $A_xB_y=B_yA_x$ for all $x,y$. In this case $\lscript ^A\circ\lscript ^B=\lscript ^{A\circ B}$.
(iii)\enspace Since
\begin{equation*}
(\iscript\circ\jscript )_{(x,y)}(\rho )=\jscript _y\paren{\iscript _x(\rho )}=T_yS_x\in\rho S_x^*T_y^*=(T_yS_x)\rho (T_yS_x)^*
\end{equation*}
we conclude that $\iscript\circ\jscript$ is a Kraus instrument with operators $T_yS_x$.
(iv)\enspace Since
\begin{align*}
(\iscript\circ\jscript )_{(x,y)}(\rho )&=\rmtr (\rho A_x)\jscript _y(\alpha _x)=\rmtr (\rho A_x)\rmtr (\alpha _xB_y)\beta _y\\
   &=\rmtr\sqbrac{\rho\rmtr (\alpha _xB_y)A_x}\beta _y=\rmtr\sqbrac{\rho C_{(x,y)}}\beta _y
\end{align*}
we conclude that $\iscript\circ\jscript$ is semitrivial with observable $C_{(x,y)}$ and states $\beta _y$. The last statement follows from
\begin{align*}
(\jscript\mid\iscript )_y(\rho )&=\sum _x(\iscript\circ\jscript )_{(x,y)}(\rho )=\rmtr\sqbrac{\rho\sum _x\rmtr (\alpha _xB_y)A_x}\beta _y\\
   &\rmtr\sqbrac{\rho\rmpost _{(\alpha ,B)}(A)_y}\beta _y=\rmtr\sqbrac{\rho\rmpost _{(\alpha ,B)}(\iscripthat\,)_y}\beta _y\qedhere
\end{align*}
\end{proof}

\begin{cor}    
\label{cor35}
If $\iscript$, $\jscript$ are trivial with observables $A$, $B$ and states $\alpha$, $\beta$, respectively, then $\iscript\circ\jscript$ is trivial with observable $C_{(x,y}=\rmtr (\alpha B_y)A_x$ and state $\beta$. Moreover, $(\jscript\mid\iscript )_y(\rho )=\rmtr (\alpha B_y)\beta$ so
$(\jscript\mid\iscript )$ is trivial with observable $\rmtr (\alpha B_y)I$ and state $\beta$.
\end{cor}

\begin{exam}  
This example illustrates that $(\iscript\circ\jscript )^\wedge\ne\iscripthat\circ\jscripthat$ except for L\"uders instruments. If $\iscript$ and $\jscript$ are Kraus instruments with operators $S_x$, $T_y$, respectively, we have seen that $(\iscript\circ\jscript )_{(x,y)}^\wedge =S_x^*T_y^*T_yS_x$. However,
\begin{equation*}
(\iscripthat\circ\jscripthat\,)_{(x,y)}=\iscripthat _x\circ\jscripthat _y=\iscripthat _x^{\,1/2}\jscripthat _y\iscripthat _x^{\,1/2}
   =(S_x^*S_x)^{1/2}T_y^*T_y(S_x^*S_x)^{1/2}
\end{equation*}
Hence, $(\iscript\circ\jscript )^\wedge\ne\iscripthat\circ\jscripthat$, in general. If $\iscript$, $\jscript$ are trivial instruments with operators $A,B$ and states $\alpha$, $\beta$, respectively, then
\begin{equation*}
(\iscripthat\circ\jscripthat\,)_{(x,y)}=\iscripthat _x\circ\jscripthat _y=A_x\circ B_y=A_x^{1/2}B_yA_x^{1/2}
\end{equation*}
However, we have seen that
\begin{equation*}
(\iscript\circ\jscript )_{(x,y)}^\wedge =\rmtr (\alpha B_y)A_x
\end{equation*}
Hence, $(\iscript\circ\jscript )^\wedge\ne\iscripthat\circ\jscripthat$, in general.\hfill\qedsymbol
\end{exam}

\begin{exam}  
We first show that $f(\lscript ^A)$ is not a L\"uders instrument and $f(\lscript ^A)\ne\lscript ^{f(A)}$ in general. To show this, we have that
\begin{equation}                
\label{eq32}
f(\lscript ^A)_y(\rho )=\lscript _{f^{-1}(y)}^A(\rho )=\sum _x\brac{A_x^{1/2}\rho A_x^{1/2}\colon f(x)=y}
\end{equation}
which is not a L\"uders instrument, in general. However, $\lscript ^{f(A)}$ is a L\"uders instrument so $f(\lscript ^A)\ne\lscript ^{f(A)}$. To be explicit we obtain
\begin{align*}
\lscript _y^{f(A)}(\rho )&=f(A)_y^{1/2}\rho f(A)_y^{1/2}=A_{f^{-1}(y)}^{1/2}\rho A_{f^{-1}(y)}^{1/2}\\
   &=\paren{\sum _x\brac{A_x\colon f(x)=y}}^{1/2}\rho\paren{\sum _x\brac{A_x\colon f(x)=y}}^{1/2}
\end{align*}
which is different than $f(\lscript ^A)$ in \eqref{eq32}. If $\iscript _x(\rho )=S_x\rho S_x^*$ is a Kraus instrument, then $f(\iscript )$ need not be a Kraus instrument. Indeed,
\begin{equation*}
f(\iscript )_y(\rho )=\iscript _{f^{-1}(y)}(\rho )=\sum _x\brac{\iscript _x(\rho )\colon f(x)=y}=\sum _x\brac{S_x\rho S_x^*\colon f(x)=y}
\end{equation*}
which is not a Kraus instrument, in general. We leave it to the reader to show that if $\iscript$ is semitrivial, then $f(\iscript )$ need not be semitrivial. However, if $\iscript (\rho )=\rmtr (\rho A_x)\alpha$ is trivial, then $f(\iscript )$ is trivial with observable $f(A)$ and state $\alpha$. Indeed,
\begin{align*}
f(\iscript )_y(\rho )&=\sum _x\brac{\iscript _x(\rho )\colon f(x)=y}=\sum _x\brac{\rmtr (\rho A_x)\alpha\colon f(x)=y}\\
   &=\rmtr\sqbrac{\rho\sum _x\brac{A_x\colon f(x)=y}}\alpha =\rmtr\sqbrac{\rho f(A)_y}\alpha\hskip 7pc\qedsymbol
\end{align*}
\end{exam}

\end{document}